\documentclass[letterpaper,11pt]{article}

\usepackage[margin=1in]{geometry}

\usepackage{microtype}

\usepackage[T1]{fontenc}
\usepackage{amsthm}
\usepackage{amsmath}
\usepackage{amssymb}

\usepackage{enumitem}

\RequirePackage{thmtools}

\declaretheorem{falseth}

\declaretheorem[sibling=falseth]{definition}
\declaretheorem[sibling=falseth]{lemma}
\declaretheorem[sibling=falseth]{remark}

\usepackage{hyperref}
\hypersetup{colorlinks=true,citecolor=blue,linkcolor=blue,urlcolor=blue}
\usepackage[nameinlink,capitalize]{cleveref}

\usepackage{algpseudocode}
\usepackage{algorithm}

\usepackage{verbatim}
\usepackage{graphicx}
\graphicspath{{./pics/}}

\newcommand{\OO}[1]{\mathcal{O} \left( #1 \right)}
\newcommand{\OMG}[1]{\Omega \left( #1 \right)}

\newcommand{\dpro}[1]{\odot_{#1}}

\newcommand{\ceil}[1]{\left\lceil #1 \right\rceil}
\newcommand{\floor}[1]{\left\lfloor #1 \right\rfloor}
\newcommand{\flfrac}[2]{\floor{\frac{#1}{#2}}}
\newcommand{\cefrac}[2]{\ceil{\frac{#1}{#2}}}

\newcommand{\bra}[1]{\left( #1 \right)}
\newcommand{\set}[1]{ \left\{ #1 \right\} }

\newcommand{\longset}[2]{ \left\{ #1 \; \middle| \; #2 \right\} }

\newcommand{\abs}[1]{\left| #1 \right|}

\title{Logarithmic Approximation for Road Pricing on Grids}

\usepackage{authblk}
\makeatletter
\renewcommand\AB@authnote[1]{\rlap{\textsuperscript{\normalfont#1}}}

\makeatother

\author{Andrei Constantinescu}
\author{Andrzej Turko}
\author{Roger Wattenhofer\thanks{\texttt{aconstantine@ethz.ch}, \texttt{andrzej.turko@gmail.com}, \texttt{wattenhofer@ethz.ch}}}
\affil{ETH Zürich, Zürich, Switzerland}
\date{}

\begin{document}

\maketitle 

\begin{abstract}
Consider a graph $G = (V, E)$ and some commuters, each specified by a tuple $(u, v, b)$ consisting of two nodes in the graph $u, v \in V$ and a non-negative real number $b$, specifying their budget. The goal is to find a pricing function $p$ of the edges of $G$ that maximizes the revenue generated by the commuters. Here, each commuter $(u, v, b)$ either pays the lowest-cost of a $u$-$v$ path under the pricing $p$, or 0, if this exceeds their budget $b$. We study this problem for the case where $G$ is a bounded-width grid graph and give a polynomial-time approximation algorithm with approximation ratio $O(\log |E|)$. Our approach combines existing ideas with new insights. Most notably, we employ a rather seldom-encountered technique that we coin under the name `assume-implement dynamic programming.' This technique involves dynamic programming where some information about the future decisions of the dynamic program is guessed in advance and `assumed' to hold, and then subsequent decisions are forced to `implement' the guess. This enables computing the cost of the current transition by using information that would normally only be available in the future.
\end{abstract}

%%%%%%%%%%%%%%%%%%%%%%%%%%%%%%%%%%%%%%%%%%%%%%%%%%%%%%%%%%%%%%%%%%%%%%%%

\section{Introduction}
A country's road network can usually be modeled as an undirected graph $G$, with cities modeled as vertices and roads connecting them as edges. The state commonly installs tolls on some roads, i.e., highways, that drivers must pay whenever traversing them. Within reason, the state would like to set tolls to maximize its revenue, but road prices can also not be too high, as this would lead to drivers looking for alternative modes of transportation and not using the road network effectively, leading to a decrease in revenue. This makes the problem particularly challenging to solve. Accounting for all aspects of the problem can also be quite demanding, as with improperly set prices, road congestion can become prohibitive, and the prices would also need to be adapted with time to account for the changing travel patterns of the drivers. In addition, perfect information about the intentions of the travelers might not be available for any given time frame --- at best an estimation based on historical data could be retrieved. Last but not least, different drivers might have different acceptance thresholds on how much they are willing to pay for any given trip. In this paper, we study a simplified version of this problem:

\begin{quote}
    We are given a graph $G = (V, E)$ and a multiset $B$ of drivers. Each driver is specified by a tuple $(u, v, b)$ consisting of two nodes in the graph $u, v \in V$ and a non-negative real number $b$, specifying the driver's budget. The goal is to find a pricing function $p : E \to \mathbb{R}_{\geq 0}$ of the edges of $G$ that maximizes the revenue collectively generated by the drivers. The revenue generated by a driver $(u, v, b) \in B$ is either the cost of a lowest-cost $u$-$v$ path under the pricing $p$, or 0, if the cost of this path exceeds the driver's budget $b$.
\end{quote}

Equivalently, the problem can be interpreted as a game between the algorithm setting the prices and the drivers: first, the algorithm observes $G$ and $B$, and selects the pricing $p$; then, each driver $(u, v, b) \in B$ computes a lowest-cost $u$-$v$ path and either pays its cost if it is at most $b$, or pays nothing otherwise. What is the maximum revenue that can be achieved by the algorithm?

So far, a variety of special cases of this problem have been investigated:
\begin{enumerate}[nosep]
    \item For $G$ being a path (\emph{``the highway problem''}), Grandoni and Rothvo{\ss} \cite{grandoni2016pricing} give a polynomial-time approximation scheme, while Briest and Krysta \cite{Briest06} show that finding the optimal revenue exactly is NP-hard.
    \item For $G$ being a tree (\emph{``the tollbooth problem''}), Gamzu and Segev \cite{Gamzu10} give a polynomial-time $\OO{\frac{\log |E|}{\log \log |E|}}$-ap\-prox\-i\-ma\-tion algorithm, while Guruswami et al.~\cite{Guruswami05} show that approximating the optimal revenue within any constant factor is NP-hard.
    \item For $G$ being a cactus, i.e., a graph whose biconnected components are either edges or cycles (generalizing trees), Turko and Byrka \cite{cactus23} give a polynomial-time $\OO{\frac{\log |E|}{\log \log |E|}}$-ap\-prox\-i\-ma\-tion algorithm, generalizing the previous result for trees.
\end{enumerate}

The case of general graphs $G$ is arguably less understood. If an approximation ratio that depends on the sizes of both $G$ and $B$ is sufficient, then a polynomial-time $O(\log |E| + \log |B|)$-ap\-prox\-i\-ma\-tion algorithm follows by the more general results of Balcan, Blum and Mansour \cite{Balcan08} on revenue maximization for unlimited-supply envy-free good pricing. Their algorithm is particularly attractive since it sets the same price to all graph edges. For completeness, we give a simpler proof tailored to our setting that a single-price algorithm can achieve an $O(\log |E| + \log |B|)$-ap\-prox\-i\-ma\-tion in \cref{sec:single_price}. On the other hand, if one is interested in an approximation ratio that only depends on the size of $G$ (which might be desirable in settings where the population size is large in comparison to the network size), then no positive results are known other than for the three graph classes outlined above. Arguably, these classes are unlikely to model real-world networks, which may contain cycles in complex patterns. 

\textbf{Our Contribution.} We study the problem on the class of grid graphs, i.e., Manhattan-like networks, and prove that for any fixed (constant) width $\omega$ of the grid, there is a polynomial-time $\OO{\log |E|}$-ap\-prox\-i\-ma\-tion algorithm for the maximum achievable revenue. While our result does not hold without the fixed width assumption, we see it as the stepping stone to understanding the complexity of various other models, e.g., bounded-pathwidth or even bounded-treewidth, for which it seems considerably more challenging. To achieve our approximation ratio, we combine previous insights with new techniques, some of which we believe could be of independent interest. The most interesting technique that we use is \emph{`assume-implement dynamic programming'}. This involves dynamic programming where some information about the future decisions of the dynamic program is guessed in advance and \emph{assumed} to hold (this enables computing the cost of the current transition by using information that would normally only be available later), and then subsequent decisions of the program are forced to \emph{implement} the guess (make it come true). This technique has scarcely appeared in previous work without an explicit highlight (e.g., implicitly in \cite{assume_implement}) and we believe deserves further popularization.

\subsection{More Related Work}

\textbf{Pricing for Envy-Free Revenue Maximization.}
The problem we study in this paper can be seen in the broader context of \emph{pricing for envy-free revenue maximization}. To explain the connection, let us quickly lay the foundations of the latter: there, one is given $k$ goods to sell and a multiset $B$ of buyers interested in acquiring subsets of the goods. Each buyer assigns a certain \emph{valuation} to each potential subset of goods they might get. Valuations can be additive or not, i.e., the value a buyer derives from getting a subset of goods need not necessarily match the sum of the values of the constituent goods. Moreover, assuming that each good has a price, the \emph{utility} that a buyer derives from getting a subset of goods is their valuation for that subset minus the total price of goods in that subset. The goal is to set the prices of the goods and return an allocation of goods to buyers maximizing the total revenue (sum of prices of sold goods) that is \emph{envy-free}, i.e., each buyer is assigned a utility-maximizing subset of goods. This problem is commonly studied in two flavors: the \emph{unit-supply} case, which we just described, and the \emph{unlimited-supply} case, where there exist infinitely many copies of each good, but each buyer can only get one copy from each good. One can also define an in-between, the \emph{limited-supply} case, where there is a known finite number of copies to sell for each good.

Armed as such, our problem can be phrased as pricing for envy-free revenue maximization as follows: in the unlimited-supply setting, consider the goods to be the edges of the graph $G$, and the buyers to be the drivers. The valuation of each driver $(u, v, b) \in B$ for a set of edges $P$ is defined as follows:
$$ v_{(u, v, b)}(P) = \begin{cases}
      b & \text{if } P \text{ defines a path between } u \text{ and } v \\
      0 & \text{otherwise}
   \end{cases} $$
   
\noindent (Note how this valuation function is not additive, and, in fact, not even monotonic.)

Envy-free pricing has been studied in a variety of settings, not only for revenue maximization, but also for social welfare maximization \cite{largeMarkets, priceDoubling, impreciseDistribution}. 
Assuming unlimited supply, Demaine et al.~\cite{demaine2008combination} proved a conditional lower bound of $\OMG{\log k}$ on the approximation ratio of any polynomial-time algorithm for revenue maximization in the general envy-free pricing problem. Their bound relies on a hardness hypothesis regarding the balanced bipartite independent set problem.
Note that this bound does not apply when restricting the problem to our specific graph-pricing setting, as can be seen from the better results cited in the previous section for $k = |E|$.

The case of limited supply (where each good is available in a certain number of copies) is also interesting and has yielded a number of polynomial-time approximation results for the maximum revenue. For the case of single-minded buyers, i.e., each buyer is only interested in a single set of goods, Cheung and Swamy \cite{Cheung08} gave a polynomial-time $\OO{\sqrt{k} \log s_{max}}$-ap\-prox\-i\-ma\-tion, where $s_{max}$ is the maximum supply of a single good. For the more specific highway problem, Grandoni and Wiese \cite{grandoni_et_al:LIPIcs:2019:11175} presented a PTAS, thus matching the aforementioned result of Grandoni and Rothvo{\ss} \cite{grandoni2016pricing} for the unlimited supply case.

\noindent \textbf{Stackelberg Pricing.} In \emph{Stackelberg games}, the \emph{leader} takes an action, and then, with knowledge of the action, the \emph{follower} replies with an action of their own. In \emph{Stackelberg pricing games}, there are $m$ items, some of whose prices are set in advance; the leader sets the prices of the remaining items, and the follower buys a minimum cost bundle subject to feasibility constraints. The leader wants to maximize the revenue, defined in terms of the \emph{priceable} items.
The \emph{Stackelberg shortest path game} (SSPG) is the previous with items being edges of a graph and the bought items forming an $s$-$t$ path. Our setup resembles a multi-follower SSPG. See \cite{stackelberg_packing,stackelberg_network_pricing,hardness_stackelberg_shortest_path,widmayer} for a survey of results relating to SSPGs. There are two essential differences between our setup and SSPGs: (i) in our problem, when the budget $b$ of a driver is exceeded, they simply no longer choose any path (so their revenue curve is discontinuous at $b$), an effect which is not captured by SSPGs;  (ii) in SSPGs some of the edge prices are supplied in advance: those edges count in the shortest path computations but not in the revenue calculation.

\subsection{Our Result and Technical Overview}

For our result, we assume that the graph $G$ is a complete $m \times \omega$ grid, where $\omega$ is considered to be fixed, and $m$ can vary as part of the problem instance. Moreover, for brevity, write $n := |B|$ for the number of drivers. Our main result is stated below:

\begin{restatable}{theorem}{gridgraph}\label{thm:gridgraph} There exists a polynomial-time approximation algorithm for the maximum revenue for the class of graphs $G$ consisting of width-$\omega$ complete grids with approximation ratio $\OO{\log m}$. 
\end{restatable}

The $\OO{\log m}$ approximation ratio is achieved as follows.
First, the algorithm partitions the drivers into $\OO{\log m}$ subsets.
Taking advantage of the additional structure of those subsets, we find a constant factor approximation of the optimal solution for each of the resulting instances (formed by the whole graph and the given driver subset).
Naturally, at least one of those driver subsets generates at least $\OMG{\frac{1}{\log m}}$ of the optimal total revenue.
Hence, one of the $\OO{\log m}$ computed price assignments yields an approximation ratio of $\OO{\log m}$ in the original instance.

Thanks to the properties of the partition of drivers into $\OO{\log m}$ subsets (denoted $B_1, B_2, \ldots$), each multiset (we allow multiple buyers with the same tuple $(u, v, b)$) $B_j$ can be further subdivided into groups, yielding independent instances of the problem
with the property that for each group there exists a row in the grid graph traversed by all drivers belonging to the group.
We reduce such instances to the so-called \emph{`rooted'} case, where all drivers' desired paths start at a single `root' vertex.
Then, the algorithm finds a constant factor approximation for such `rooted' instances using dynamic programming.
We combine the results of the rooted instances to obtain a final price assignment achieving a constant factor approximation for the instance consisting only of drivers $B_j$.
Then, among the $\OO{\log m}$ price assignments, we choose one with the highest revenue, arriving at an $\OO{\log m}$ approximation for the initial instance.

The `assume-implement' dynamic programming technique is showcased in \cref{sec:rooted}.
Normally, in a dynamic program for an optimization problem, in each state one considers
solutions to corresponding subproblems and chooses the one maximizing a certain objective function.
This becomes more complicated when the objective function depends not only on the solution to the subproblem
but also on how the solution is constructed in the remaining part of the global instance.
The essence of the `assume-implement' technique is to assume some properties of the solution
outside the subproblem that are necessary to compute and maximize the objective function inside it. These assumptions are then lazily implemented when solving larger subproblems. In particular, the assumptions are included in the state description so that they can be taken into consideration when a solution to a subproblem $S$ is used to solve a larger subproblem $S'$ containing $S$. The algorithm would ensure that the assumptions are implemented either directly in $S'$ or
with the help of making certain new assumptions on the solution outside of $S'$, which are included in the state description for $S'$ and their implementation delegated to even larger subproblems.
This way, the algorithm can lazily implement assumptions on some properties of the solution.

One interesting contribution concerns grid graph compression (\cref{lemma:compression}). We show that, given any weighted width-$\omega$ grid, another weighted width-$\omega$ grid of length bounded by a function of $\omega$ exists such that for any $1 \leq i \leq j \leq \omega$ the distance between the $i$-th and the $j$-th vertex in the first row is the same in the two grids.
The proof is based on a generally useful fact that we prove: given a set of vertices in a general weighted graph, we can find a collection of all-pairs-shortest-paths between them inducing a subgraph with a convenient shape, i.e., a limited number of vertices of degree three or more. We could not identify any references of this fact.

\section{Rooted Case}\label{sect:rooted}
\label{sec:rooted}

\newcommand{\Crounding}{4}

In this section, we define a simpler variant of our problem and present a polynomial-time constant factor approximation algorithm for it, which will be a building block of the main algorithm.

\begin{definition}
    An instance of our problem on a grid graph $G$ with driver multiset $B$ is \emph{rooted} if one of the vertices in the top row of $G$ (the \emph{root}, denoted $r$) is an endpoint of every path desired by drivers in the multiset $B$.
\end{definition}

We will prove the following:

\begin{lemma}
\label{rooted_rev}
    Assume a fixed grid width $\omega$. For any rooted instance with maximal revenue $\mathrm{OPT}_r$, a~price assignment generating at least $\frac{\mathrm{OPT}_r}{\Crounding}$ in revenue can be found in polynomial time.
\end{lemma}

\subsection{Rounding}
\label{sect:rounding}

One of the crucial techniques used throughout the algorithm is price rounding.
Let us formalize the underlying observation.

\newcommand{\roundedSet}{\mathcal{P}}

\begin{lemma}{\textbf{(Rounding)}}
    \label{rounding_lemma}
For any instance ($G, B$) of the problem on a grid,
there exists a pricing of edges which results in revenue of at least $\frac{1}{\Crounding}$ of the optimal one such that the price of each edge belongs to the set:
$$ \roundedSet = \longset{\frac{b_{max}}{2^t}}{t \in \set{0, 1, \dots, \ceil{\log \bra{ 4 m \cdot \abs{B} }}}} \cup \set{0} $$
where $b_{max} = \max_{(u, v, b) \in B} b$.
\end{lemma}
\begin{proof}
	Consider any price assignment $p$ generating the optimal revenue $\mathrm{OPT}$.
	If a price of any edge is greater than $b_{max}$, it can be lowered to $b_{max}$ without loss of revenue.
	We further round the price of each edge down to the nearest value from $P$.
    Let us consider any edge $e$, its initial price $p_e$ and the rounded price $p'_e$.
    The price of some edges can be rounded down to $0$, but if $p'_e = 0$, then $p_e < \frac{b_{max}}{4 m \cdot \abs{B}}$.
    Otherwise, we are guaranteed $p'_e \geq \frac{1}{2} p_e$.
    Thus: \[p_e \geq p'_e \geq \frac{p_e}{2} - \frac{b_{max}}{4 m \cdot \abs{B}}\]

    Let $d$ and $d'$ be the distances (costs of a cheapest path) between any two vertices in $G$ with respect to the initial and rounded prices, respectively. 
    We have: $d \geq d' \geq \frac{d}{2} - \frac{b_{max}}{4 \cdot \abs{B}}$.
    As the prices only decrease, no driver will be priced-out of the market by the rounding and the same inequality naturally holds for each drivers' contribution to the revenue.
	Summing these inequalities over all drivers results in a lower bound on the revenue of $\frac{\mathrm{OPT}}{2} - \frac{b_{max}}{4}$.
	As, naturally, $\mathrm{OPT} \geq b_{max}$, this concludes the proof.
\end{proof}

In the above proof we chose the constant scaling factor of $2$ for the sake of clarity and ease of exposition.
It can be replaced with $1 + \epsilon$ for any constant $\epsilon > 0$.
This, together with adjusting the smallest element in $\roundedSet$, would allow us to prove \cref{rounding_lemma} for any constant $c > 1$ instead of $4$.

By this lemma we can restrict our attention to prices from $\mathcal{P}$.
An important consequence of this is that the costs of paths in $G$ will always be multiples of
$p_{min} = \frac{b_{max}}{2^{\ceil{\log \bra{4 m \cdot \abs{B}}}}}$.
Another observation is that once a distance between any two vertices in $G$ exceeds $b_{max}$, 
it does not matter how much it is exactly, because no driver will ever choose this path anyway.
Thus, we will unify all paths with total cost exceeding $b_{max}$ and treat them as if they had cost $\infty$.
Hence, we can restrict our attention to distances that are $\infty$, $0$ or multiples of $p_{min}$ not exceeding $b_{max}$.
This yields a polynomial upper bound on the number of possible distances between two vertices of $G$,
which allows us to use distances as state descriptions in the dynamic programming.

\subsection{Assume-Implement Dynamic Programming}

Before we proceed to the description of the algorithm, assume we index rows of $G$ top to bottom and introduce some notation:
\begin{itemize}[nosep]
\item $R_i$ -- the set of vertices in the $i$-th row of $G$.
\item $G_i$ -- the subgraph of $G$ consisting of vertices from rows $i$ and $i+1$ and edges in the $i$-th row and between rows $i$ and $i+1$.
\item $B_{R_i}$ -- the multiset of drivers whose desired paths start in $R_i$.
\end{itemize}

Let us consider a subgrid $H$ consisting of some consecutive rows of $G$.
For the purpose of paths that only pass through $H$ (have endpoints either outide of $H$ or on its boundary),
we do not need to consider all the edges inside $H$.
It is enough to consider the distances between the vertices on the boundary of $H$.
In some parts of the algorithm we will assume these distances to be equal to certain values and
then construct price assignments inside $H$ that fulfill those assumptions.
The following concept allows us to formalize the assumptions on pairwise distances for a set of vertices.

\begin{definition}
    \label{def:distance_matrix}
    Given a grid $H$ and a subset $S$ of its vertices,
    a \emph{distance matrix} for $S$ is an $\abs{S} \times \abs{S}$ matrix $D$ indexed by vertices from $S$,
    where $D_{a, b}$ is the assumed distance, i.e. sum of weights along a cheapest path, between $a$ and $b$ in $H$.

    We will say that a certain edge-weighting \emph{realizes} a distance matrix if it satisfies the above condition
    (for all $a, b \in S$, $D_{a,b}$ is actually equal to the distance between $a$ and $b$ in $H$).

    \textbf{Note:}~We will frequently consider $H$ to be a subgraph of the (weighted) grid $G$.
    In that case the distance matrix describes the cheapest paths in $H$ and not $G$.
    Even though under some weight assignments $G$ may contain shorter paths between vertices in $S$.
\end{definition}

In the algorithm we will often have to merge two graphs into a single one.
We will be interested in knowing the distances in the union of the two graphs
based on assumptions about the distances in the original graphs.
The following formalizes this operation.

\begin{definition}
\label{def:dist_product}
Let $G_A$ and $G_B$ be two same-width edge-disjoint weighted grids with (possibly) some edges missing that share a common row with vertex set $R$ that is the top row of one and the bottom row of the other.
Let $A$ and $B$ be distance matrices for vertices $S_A$ in graph $G_A$ and $S_B$ in $G_B$, such that $R \subseteq S_A \cap S_B$.
Then, the \textbf{product} of $A$ and $B$, denoted $A \dpro{} B$, is a distance matrix $C$ for $S_A \cup S_B$ that
describes the distances in (is realized by) the union of $G_A$ and $G_B$
under the assumption that $A$ and $B$ are realized.

Sometimes, in order to limit the dimensionality, we will want the product of $A$ and $B$
to be defined only for a certain subset $S \subseteq S_A \cup S_B$, in which case we write $A \dpro{S} B$.
\end{definition}

Note that the product of two distance matrices is well-defined:
because $R$ is a separator of $G_A$ and $G_B$, any shortest path between two vertices from $S_A \cup S_B$ can be partitioned into segments whose lengths are determined by either $A$ or $B$.
Thus, $A$ and $B$ already determine all pairwise distances in $S_A \cup S_B$ in $G_A \cup G_B$.

Each state of the dynamic program corresponds to a certain ($i$-th) row of the grid
and two distance matrices (\cref{def:distance_matrix}):
\begin{itemize}[nosep]
\item \emph{lower matrix} $L$ --
a distance matrix for $R_i$ in the part of $G$ below and including row $i$,
\item \emph{upper matrix} $U$ --
a distance matrix for $R_i \cup {r}$ in the grid above row $i$ excluding edges in of the $i$-th row.
This is the matrix on which we use the assume-implement technique.
\end{itemize}
Thanks to the rounding technique, we can restrict out attention to distance matrices
with values equal to $0$, $\infty$ or a multiple of $p_{min}$ as argued in \cref{sect:rounding}.

The value of $dp_{i, L, U}$ is defined as the maximum revenue generated by the drivers
whose non-root endpoint lies on or below row $i$ under a price assignment
to the edges in $i$-th row and below that realizes the matrix $L$ and assuming that
the price assignment to the edges above row $i$ realizes the upper matrix $U$.
If the lower matrix is not feasible (there does not exist a weight assignment realizing it),
we set $dp_{i, L, U} := -\infty$.
Note that we do not treat the upper matrix this way, because
for the sake of this particular state of the dynamic program we assume that a price assignment realizing it exists.
Correctness is guaranteed since this assumption is contained in the description of the state and checked at every step.

The values of $dp_{i, L, U}$ are computed bottom-up and naively, i.e. for each state we consider
all possible states from the previous row and all possible price assignments $w: E(G_i) \rightarrow \roundedSet$.
That is, we check all ways of pricing the edges in $i$-th row and those between the current and the previous row.
At each level we include the revenue from the drivers whose desired paths start in $R_i$ ($B_{R_i}$).
For convenience, we write weight assignments as distance matrices for $R_i \cup R_{i+1}$.
Let $\mathcal{W}_i$ denote the set of such distance matrices realized by $G_i$ with all assignments $w: E(G_i) \rightarrow \roundedSet$.
Then, $dp_{i, L, U}$ equals:%
\begin{equation}
\label{eq:dp}
dp_{i, L, U} \hspace{0.27cm} := \hspace{-0.75cm} 
\max_{
    \begin{array}{c} W \in \mathcal{W}_i, L', U', \; s.t. \\
    L' \dpro{R_i} W = L \\
    U \dpro{R_{i+1} \cup \set{r}} W = U' \end{array}} \hspace{-0.95cm}
    dp_{i+1, L', U'}
    \hspace{0.45cm} + \hspace{-0.55cm}
    \sum_{
    \begin{array}{c} (v, r, b) \in B_{R_i}, \; s.t. \\
    b \geq \bra{L \dpro{} U}_{r, v} \end{array}}
    \hspace{-0.95cm}\bra{L \dpro{} U}_{r, v}
\end{equation}

\begin{figure}[t]
    \centering
    \includegraphics[width=0.65\textwidth]{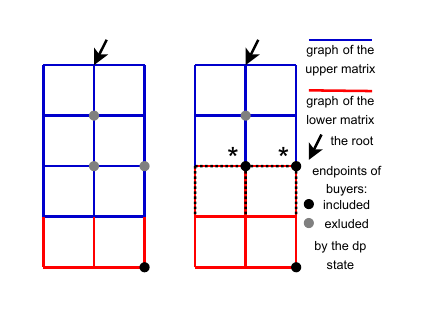}
    \caption{
        Dynamic programming states for row $i+1$ (left) and $i$ (right).
        For reference, here $i = 2$.
        Black dotted edges represent the graph $G_i$.
        Black vertices marked with asterisks are
        non-root endpoints for the drivers in $B_{R_i}$ -- those are exactly the drivers from the sum in \cref{eq:dp}.
    }

\end{figure}

The left term denotes the best revenue that can be achieved by the
drivers from rows below $i$ assuming the distance matrices $L$ and $U$ are realized.
It is a maximum over all possible lower and upper distance matrices
$L'$ and $U'$ for the previous row and each possible way $W$ to price the edges in $G_i$.
Here we only consider the combinations that are consistent with distance matrices $L$ and $U$:
% , i.e.:
\begin{itemize}[nosep]
\item $L' \dpro{R_i} W = L$ : price assignments to edges from rows $i+1$ and below realizing the matrix $L'$ extended by $W$ realize the matrix~$L$,
\item $U \dpro{R_{i+1} \cup \set{r}} W = U'$ : price assignments to edges above row $i$ realizing $U$ extended by $W$ realize the matrix~$U'$.
\end{itemize}
% \ajnote{The following two paragraphs are new.}
Note that $E(G_i) = E(L) \setminus E(L')$ and $E(G_i) = E(U') \setminus E(U)$, where $E(M)$ denotes the edges of the graph corresponding to the matrix $M$.
Thus, it is clear that by pricing $E(G_i)$ we extend the lower matrix $L'$ to $L$ and the upper matrix $U$ to $U'$.
In other words, as we move up through the rows the graph for the lower matrix, for which the solution is already constructed, grows.
On the other hand, the graph for the upper matrix shrinks, which highlights the ``implement'' aspect of the assume-implement technique:
the solution for the upper part of the grid that has been assumed (by $U'$) in a state for the previous row
is partially implemented by the price assignment to $E(G_i)$ while the remaining part still remains assumed by $U$.

The right term is the revenue from the drivers in $B_{R_i}$
(each driver contributes the distance between her endpoints provided it does not exceed her budget).
Note that $U \dpro{} \mathit{L}$ is a distance matrix describing the whole grid $G$ with respect to $R_i \cup \set{r}$, so
for each $v \in R_i$ we know that $\bra{L \dpro{} U}_{r, v}$ is the distance from $r$ to $v$ in $G$ under any weight assignment
realizing $L$ and $U$.
This highlights the other aspect of the assume-implement technique ---
in order to calculate the revenue of the drivers in the current row, we need knowledge of
the prices of the edges in row $i$ and below (the lower matrix), which are already constructed,
as well as the price of edges above row $i$ (the upper matrix), which are not part of the solution
for the current state of the dynamic program.
The lower matrix is already known because the solution is constructed bottom-up.
The upper matrix is assumed (the ``assume'' aspect of the assume-implement technique), so that we have
all the necessary information to process drivers in the current row.

The values of $dp$ for the bottom row are initialized to the revenue from the respective
drivers for all upper matrices and feasible lower matrices as follows:
\[ dp_{n, L, U} := \sum_{
    \begin{array}{c} (v, r, b) \in B_{R_n}, \; s.t. \\
    b \geq \bra{L \dpro{} U}_{r, v} \end{array}}
    \bra{L \dpro{} U}_{r, v} \]
A lower matrix is feasible if there exists a weight assignment realizing it.
In this case, the lower matrix is describing a single row of the grid,
so this condition can easily be checked by naively iterating over all weight assignments
to the edges inside $R_n$ ($\abs{R_n} = \omega$).

The algorithm obtains the revenue $\max_{L} dp_{1, L, U_{\infty}}$, i.e.
chooses the biggest revenue among the lower matrices and the
infinite upper matrix $U_{\infty}$ which has $0$'s on the main diagonal and $\infty$'s elsewhere.
It is because the upper matrices for the top row by definition describe an empty graph.
$\max_{L} dp_{1, L, U_{\infty}}$ is also the optimal revenue under the assumption that all edge prices are rounded (belong to $\roundedSet$).
This follows from the definition of the $dp$ values and the fact that for each rounded price assignment there exists a distance matrix realized by it.
Thus, by the Rounding Lemma, $\max_{L} dp_{1, U_{\infty}, L} \geq \frac{\mathrm{OPT}_r}{\Crounding}$.

The price assignment itself can retrieved by a simple backtracking procedure, because
the algorithm stores $W$ and the best previous state from \cref{eq:dp}
alongside the values of $dp$.

\begin{lemma}
    The above algorithm runs in polynomial time.
\end{lemma}
\begin{proof}
    The number of states in the dynamic programming is bounded by a product of $m$ (the number of rows in $G$) and the numbers of possible lower and upper matrices.
    The upper and lower matrices are respectively of size $(\omega + 1) \times (\omega + 1)$, and $\omega \times \omega$,
    both of which are constant because we are working on a bounded-width grid.
    As discussed in \cref{sect:rounding}, the number of possible distance values is polynomially bounded.
    Hence, the number of different matrices and thus the number of states is polynomially bounded.

    Calculating the value of each state is also done in polynomial time.
    As we have already shown, the number of different matrices $U'$ and $L'$ is polynomially bounded.
    $\abs{\mathcal{W}_i}$ is also polynomially bounded, because we assign a polynomial number of weights (from $\roundedSet$) to a constant ($2\omega - 1$) number of edges.
    The same reasoning applies to the initializaton of the $dp$ for the bottom row.

    Since the number of states is polynomially bounded and the time to calculate the value of each state is polynomial, the algorithm runs in polynomial time.
\end{proof}

The above algorithm can be easily extended to work with incomplete grids (i.e. grids where some of the edges are missing).
Such a grid can be modeled as a complete grid, where we require some of the edges to have infinite weight ($b_{max}+1$).
Because no driver will be able to afford such edges, they will never be used in the produced solution -- as if they were not present.
On the other hand, the remaining edges are priced and sold as usual.
We enforce this condition by changing the $\mathcal{W}_i$ set from \cref{eq:dp} to model
price assignments $w: E(G_i) \rightarrow \roundedSet \cup \set{\infty}$ such that all edges not present in the incomplete grid are assigned $\infty$.
Note that now $\mathcal{W}_i$ sets will be different for particular rows because the set of missing edges is different for each row.
Both correctness and polynomial time complexity follow from the previous discussion.

\section{Decomposition}
\newcommand{\MaxGadgetDepth}{\frac{\omega^5}{4}}
\newcommand{\TwoMaxGadgetDepth}{\frac{\omega^5}{2}}
\newcommand{\FourMaxGadgetDepth}{\omega^5}

\label{sec:decomp}

The algorithm first partitions the drivers using a recursive decomposition of the grid $G$.
At each level of the decomposition, the grid is split into blocks, which are $\omega$-wide subgrids (continuous subsequences of rows).
At the first level $\mathcal{L}_1$, the decomposition is trivial: the whole grid $G$ forms a single block.
At each of the following levels $j+1$, $\mathcal{L}_{j+1}$ is obtained by splitting each block of $\mathcal{L}_j$ into at most two blocks.
Let $H$ be a block in $\mathcal{L}_j$.
The \emph{middle row} of vertices in $H$ will divide the block into two:
the upper block will be formed by all rows above the middle one
and the lower block will be formed by all rows below it.
Provided that those blocks are non-empty, they are added to $\mathcal{L}_{j+1}$.
For reasons that we describe later, the process finishes at the last level where all blocks have length at least
$\TwoMaxGadgetDepth$.
This gives us $\log m + 1$ levels.

Note that, although the blocks at a single level of the decomposition do not necessarily cover the whole grid, they are disjoint.
Also, all blocks across all $\mathcal{L}_j$'s form a laminar family, i.e., any two blocks are either disjoint or one is fully contained in the other.
Furthermore, at the last level of the decomposition, the blocks are of size at most $\FourMaxGadgetDepth + 1$.
Note that at least one block at the last level has to be shorter than that.
Otherwise, we could create the next level of the decomposition with all blocks of length at least $\TwoMaxGadgetDepth$.
The existence of a block of length at most $\FourMaxGadgetDepth$ implies that all blocks are of length at most $\FourMaxGadgetDepth + 1$
because of the following observation.

\begin{remark}
Block lengths at the same level of the decomposition differ by at most one.
\end{remark}
\begin{proof}
At the first level, we only have one block, so the claim holds naturally.
At each of the following levels, assuming that we have blocks of sizes belonging to the set $\set{s, s+1}$ for some $s$,
in the next level we will have blocks of sizes ranging from $\cefrac{s}{2}$ to $\flfrac{s-1}{2}$.
It is because the first one is split into blocks of sizes $\cefrac{s}{2}$ and $\flfrac{s}{2}$,
whereas the second one is split into blocks of sizes $\cefrac{s-1}{2}$ and $\flfrac{s-1}{2}$.
For odd $s$ the block sizes are $\frac{s+1}{2}$ and $\cefrac{s-1}{2}$, which differ by one.
For even $s$ they are $\frac{s}{2}$ and $\frac{s-2}{2}$, which also differ by one.
Hence, lengths blocks on the next level also differ by at most one.
\end{proof}

\begin{definition}
A driver $(u,v,b)$ is \emph{assigned} to a block $H$ iff $H$ is the smallest block containing vertices $u$ and $v$.
We will denote a multiset of those drivers as $B_H$.
\end{definition}

For each driver such a block is unique (from laminarity) and must exist because the whole grid is one of the blocks.
Also, for a driver $(u, v, b)$ assigned to $H$, $u$ and $v$ must lie on opposite sides of the \emph{middle row} of $H$
(or on the middle row itself).

All drivers assigned to blocks from $\mathcal{L}_j$ form $B_j$.
From now on we focus on a single level of the decomposition $j$ and present a constant factor approximation algorithm for the problem restricted to $B_j$.
In other words, we will arrive at a price assignment to all edges in $G$ whose revenue will be within a constant factor of the optimal revenue with respect to $B_j$.
Let us denote the latter value by $\mathrm{OPT}_j$.
Because we partition the drivers into $\log m$ multisets, one of the $\mathrm{OPT}_j$'s must be at least $\frac{1}{\log m}$ of the optimal global revenue.
By picking the best of the $\log m$ solutions, we will achieve a $\OO{\log m}$-approximation globally.

\subsection{Separating Subproblems}

Intuitively, we would like to split the problem restricted to $B_j$ into subproblems corresponding to the blocks of $\mathcal{L}_j$ (fragmentation of $G$)
and drivers assigned to them (partition of $B_j$).
Ideally, we would solve those subproblems independently and each of the solutions would yield a pricing of edges in the corresponding block
achieving a constant factor approximation of the revenue with respect to the drivers assigned to that block.
Then, we could apply all those solutions simultaneously to the whole grid and gain a constant factor approximation with respect to $B_j$.

This, however, is not possible as the blocks of $\mathcal{L}_j$ together with assigned drivers do not form independent subproblems.
Although, for any $H$, both endpoints of all paths desired by a driver in $B_H$ are in block $H$,
a cheapest path desired by her may contain edges outside of $H$.
The solution is to extend the subproblem corresponding to $H$ to include some of the edges outside of $H$ and optimize their prices
with respect to the revenue from $B_H$.
If we did this naively and added all edges that could be used by drivers assigned to $H$,
we would have to extend the corresponding subproblem to the whole grid.
Then, we would be able to optimize for only one of the blocks of $\mathcal{L}_j$ at a time,
which is not helpful, as $\abs{\mathcal{L}_j}$ can be as large as $n$.

However, it turns out that extending a subproblem by adding only a constant number of rows adjacent to the corresponding block is enough.
This is because we can model any finite-length grid using a relatively shallow grid with the same width.
In other words, for any block $H$, it makes no difference to the revenue whether the drivers from $B_H$ can use all the edges in $G$
or are restricted to $H$ and a constant number of rows above and below it.
Let us formalize this observation.

    \begin{restatable}{lemma}{compressionlemma}\label{lemma:compression}
    For any weighted grid of width $\omega$ there exists another weighted grid of the same width and depth at most $\MaxGadgetDepth$
    such that the distances between vertices in the first row of the original grid are preserved.
    \end{restatable}

    For the sake of brevity, here we only present the main idea of the proof, while the full proof can be found in \cref{apx:compression_lemma}.
    First, we identify a partial grid that has a simple structure, but maintains the distances between the vertices in the first row.
    We realize this by finding a collection of all-pairs shortest paths between vertices in the first row which results
    in few \emph{crossing vertices}, i.e., vertices where two paths from the collection join or diverge.
    Only those vertices can have a degree greater than two in the graph induced by the collection of shortest paths.
    Then, we compress areas of the partial grid that contain only vertices of degree two.
    Because there is only a constant number of vertices of degree three or more, we arrive at the desired result.
    The existence of the aforementioned collection of shortest paths is guaranteed by the following lemma.

    \begin{restatable}{lemma}{crossingvtexlemma}
    \label{lemma:crossing_vertices}
    For any weighted graph $G$ and a subset $S$ of its vertices, there exists a collection
    of shortest paths in $G$ between all pairs of vertices from $S$ that results in at most
    ${\abs{S} \choose 2} \cdot \bra{{\abs{S} \choose 2} - 1}$ crossing vertices outside $S$.
    \end{restatable}

    The proof of the above, which can also be found in \cref{apx:compression_lemma}, is based on the observation that
    when two paths have multiple crossing vertices, one of them can be rerouted so that the number of crossing vertices is limited to two.
    Unfortunately, because of the lack of a natural monovariant that would guarantee termination, instead
    of rerouting the paths greedily, the proof
    reroutes the paths in a specific order and maintains a more complex invariant to ensure that the
    number of crossing vertices is limited globally.

Based on \cref{lemma:compression}, we create independent subproblems using an odd- and even-step approach.
Let us index the blocks of $\mathcal{L}_j$ from top to bottom.
In the odd-step (even-step), we extend all the odd-indexed (even-indexed) blocks by $\MaxGadgetDepth$ up and down.
Those extra rows are assigned to the corresponding blocks, i.e. are part of the corresponding subproblems and will be priced according to their solution.
Because of that, they cannot overlap with other blocks in the same step or their extensions.
And indeed they do not as all blocks have length at least $\TwoMaxGadgetDepth$.
Thich means that blocks of the same parity are at least $\TwoMaxGadgetDepth$ apart---enough to accomodate two
extensions.

Note that some of the extensions could be shorter than $\MaxGadgetDepth$.
This can only happen to the top and bottom blocks which are closer than $\MaxGadgetDepth$ to the top or bottom end of the grid.
Let $H$ be such a top (bottom) block.
The extension of $H$ to the top (bottom) is shorter than $\MaxGadgetDepth$, but it reaches the end of the grid.
Thus, it contains all the paths to the top (bottom) of $H$ that drivers assigned to $H$ could use.

Each of the steps (odd and even) consists of independent and disjoint subproblems,
which are solved by the algorithm presented in \cref{sec:single_block}.
We ensure the separation of the subproblems within the same step by pricing the edges between them at $\infty$ (more specifically $b_{max} + 1$).
Thus, solutions to those subproblems can be applied simultaneously forming a global solution to the whole grid.
Because multisets of drivers assigned to different subproblems form a partition of $B_j$,
one of the subproblems must have an optimal solution yielding at least $\frac{\mathrm{OPT}_j}{2}$ revenue.
As for each subproblem separately we will achieve a constant factor approximation,
one of the solutions has to generate revenue within a constant factor of $\frac{\mathrm{OPT}_j}{2}$.

All subproblems for blocks in all levels of decomposition apart from the last one
have a special structure that we exploit in the algorithm.
Because such a block $H$ is divided into two blocks by the middle row,
only drivers ($u, v, b$) whose vertices $u$ and $v$ are separated by or lying on the middle row of $H$ can be assigned to $H$.
Otherwise, they would be assigned to the upper or lower block (in the next level of the decomposition).

That does not hold for a block $H'$ at the last level of the decomposition.
However, as we observed earlier, such a block $H'$ is at most $\FourMaxGadgetDepth + 1$ long.
Even with extensions, it gives us a constant upper bound on the size of the instance associated with $H'$
($\frac{3}{2} \FourMaxGadgetDepth + 1$)
and, thus, a constant upper bound on the number of edges in that instance.
This observation, together with the Rounding Lemma (\cref{rounding_lemma}), allows us to
find a $4$-approximation for $H'$ with extensions and drivers $B_{H'}$ in polynomial time.
We iterate over all price assignments which draw prices from the set $\roundedSet$ from the Rounding Lemma,
and apply them to the edges in the instance associated with $H'$.
Since $\abs{\mathcal{P}}$ is polynomial in the size of the instance and there is a constant number of edges,
such a procedure is polynomial.
By the Rounding Lemma, one of those assignments will achieve revenue of at least $\frac{1}{\Crounding}$ of the optimal one.
From now on, we focus on the subproblems for blocks in all levels of decomposition excluding the last one.

\section{Algorithm for a Single Block}
\label{sec:single_block}
\newcommand{\locrevab}[2]{\mathrm{#1}^{\mathrm{#2}}_{H, a, b}}
\newcommand{\locreva}[2]{\mathrm{#1}^{\mathrm{#2}}_{H, a}}
\newcommand{\locrev}[2]{\mathrm{#1}^{\mathrm{#2}}_{H}}

In this section, we focus on a single block $H$ from a certain level of the decomposition $\mathcal{L}_j$.
Blocks on the last level of the decomposition will be solved naively, as described in \cref{sec:decomp}, so we do not concern ourselves with them here. We consider the instance of our problem consisting of the extended block of $H$, denoted $H_{\mathrm{ext}}$, and all drivers assigned to it $B_H$. For this instance, we will present a poly-time
constant factor approximation algorithm, i.e., an algorithm for assigning prices to the edges of $H_{\mathrm{ext}}$ such that the revenue generated by drivers in $B_H$ is a constant-factor away from optimal.

To begin, note that, by construction of the decomposition, for each driver $(u, v, b) \in B_H$, we have that one of $u$ and $v$ lies on or above the middle row of $H$ and the other on or below the middle row of $H$. As a result, each $u$-$v$ path contains at least one vertex from the middle row of $H$. Let us assume without loss of generality that $v$'s row is not above $u$'s row to eliminate a
case distinction. 

Consider any $u$-$v$ path and let $x$ be the first vertex from the middle row of $H$ on the path and $y$ be the last vertex from the middle row of $H$ on the path. Then, the path can be split into three sub-paths $u$-$x$, $x$-$y$ and $y$-$v$, which we call the \emph{upper section}, \emph{middle section} and \emph{lower section} of the path, respectively. Note that any of the three sections could be trivial paths with no edges.

\begin{figure}[t]
    \centering
    \includegraphics[width=0.5\textwidth]{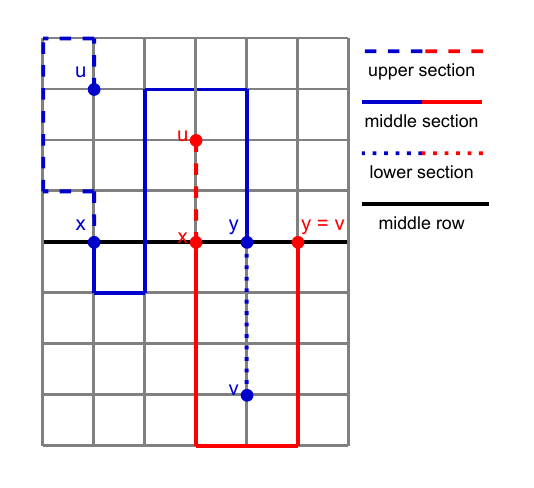}
    \caption{
        Example of two paths in a block $H$ desired by two different drivers (red and blue) and their partition into sections.
        Note that the lower section of the red path is empty, because the corresponding endpoint lies on the middle row, so, by definition, $y = v$.
    }
\end{figure}

Consider an optimal solution $S^*$ for our instance. For each driver $(u, v, b) \in B_H$, there could be multiple lowest-cost $u$-$v$ paths under the pricing $S^*$ --- choose one such path for each driver. Let $M$ be the set of vertices on the middle row of $H$. For any two vertices $s, t \in M$, write $B_{s, t}^*$ for the multiset of drivers in $B_H$ whose chosen path has $(x, y) = (s, t)$ and $R_{S^*}(B_{s, t}^*)$ for the total revenue generated under the pricing $S^*$ by drivers in $B_{s, t}^*$. Then, we can write the optimal revenue as  $R_{S^*} = \sum_{x, y \in M} R_{S^*}(B_{s, t}^*)$.
As a result, there must exist $s, t \in M$
such that $R_{S^*}(B_{s, t}^*) \geq \frac{1}{\omega^2}R_{S^*}$. This is because there are $\omega^2$ terms in the summation, meaning the average of the terms is $\frac{1}{\omega^2}R_{S^*}$ --- at least one term is no lower than the average, implying the previous.
This brings us to the crucial idea:
if we could generate in polynomial time solutions $S_{s, t}$ for all $s, t \in M$ such that the revenues generated by them satisfy
$R_{S_{s, t}} \geq \alpha R_{S^*}(B_{s, t}^*)$ for some fixed constant $\alpha > 0$,
then the best of these solutions will give an $\frac{1}{\alpha} \omega^2$-approximation of the optimal revenue.
This will be the approach that we will be taking.

From now on, consider two fixed vertices on the middle row $s, t \in M$. We want to construct a solution $S_{s, t}$ generating revenue at least $\alpha R_{S^*}(B_{s, t}^*)$. 
To achieve this, let us more closely analyze the paths of drivers in $B_{s, t}^*$: all such paths consist of an upper section that ends at $s$, a middle section between $s$ and $t$, and a lower section starting at $t$. The upper section can only use edges with one endpoint above the middle row and the other endpoint either at $s$ or also above the middle row (marked blue in \cref{fig:pricing_scheme}). Analogous considerations apply to the lower section (edges marked red in \cref{fig:pricing_scheme}). The middle section might be different for different paths, but, since all paths are lowest-cost, it must have exactly the same total price in all paths under consideration; call it $p_{s, t}^\mathrm{mid}$. With these observations, let us construct another solution $S'_{s, t}$ by starting with $S^*$ and modifying the prices as follows:
\begin{enumerate}[nosep]
    \item Set to $\infty$ the prices of all edges going up from the middle row, except the one exiting $s$, which keeps its price.
    \item Set to $\infty$ the prices of all edges going down from the middle row, except the one exiting $t$, which keeps its price.
    \item Set to $\infty$ the price of all edges on the middle row, except those in the $s$-$t$ range.
    \item Reprice the edges in the $s$-$t$ range on the middle row to sum to $p_{s, t}^\mathrm{mid}$. This can be achieved by making all but one of them zero, except in the case $s = t$, where $p_{s, t}^\mathrm{mid} = 0$ anyway. 
\end{enumerate}
Edges from steps 1-3 are marked with black dotted lines in \cref{fig:pricing_scheme}, whereas the ones from step 4 are marked green.

\begin{figure}[t]
    \centering
    \includegraphics[width=0.5\textwidth]{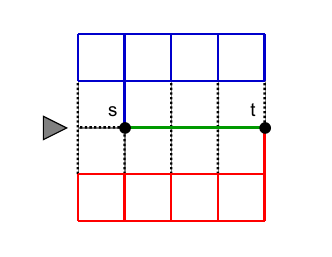}
    \caption{
        A scheme of the partition of edges of $H_{\mathrm{ext}}$.
        The middle row is marked with a triangle.
        Black dotted edges are priced at $\infty$ in all $S^i_{s, t}$.
        The red edges are used to optimize for the lower section ($S^{\mathrm{low}}_{s,t}$), the blue edges for the upper section ($S^{\mathrm{up}}_{s,t}$), and the green edges for the middle section ($S^{\mathrm{mid}}_{s,t}$).
        Whenever one of those groups is used to optimize revenue, the other two
        are priced to $0$.
    }
    \label{fig:pricing_scheme}
\end{figure}

By the previous observations, we have that $R_{S_{s, t}'}(B_{s, t}^*) = R_{S^*}(B_{s, t}^*)$, i.e., drivers in $B_{s, t}^*$ generate exactly the same revenue under $S^*$ and $S'_{s, t}$. Consequently, it would be enough to find a solution $S_{s, t}$ such that $R_{S_{s, t}} \geq \alpha R_{S_{s, t}'}(B_{s, t}^*)$. To do this, note moreover that $R_{S_{s, t}'}(B_{s, t}^*)$ can be written as $R_{S_{s, t}'}^\mathrm{up}(B_{s, t}^*) + R_{S_{s, t}'}^\mathrm{mid}(B_{s, t}^*) + R_{S_{s, t}'}^\mathrm{low}(B_{s, t}^*)$, where the three summands correspond to the revenues generated by the upper, middle, and respectively lower sections of the paths that drivers in $B_{s, t}^*$ will choose under the pricing $S_{s, t}'$.
Hence, at least one of the three summands is at least $\frac{1}{3} R_{S_{s, t}'}(B_{s, t}^*)$, leading us to the next idea: if we could construct three solutions $S_{s, t}^i$ for $i \in \{\mathrm{up}, \mathrm{mid}, \mathrm{low}\}$ such that $R_{S_{s, t}^i} \geq \beta R_{S_{s, t}}^i (B_{s, t}^*)$ for some fixed $\beta > 0$, then taking $S_{s, t}$ to be the best of them achieves our goal for $\alpha = \frac{1}{3} \beta$. In the following, we explain how to do this for each $i$.
For convenience, we will mention the colors of particular edge groups in \cref{fig:pricing_scheme}.

\begin{enumerate}[nosep]
    \item For $i = \mathrm{up}$, note that this is symmetric with $i = \mathrm{low}$, which we treat below.
    \item For $i = \mathrm{mid}$, note that a revenue of at least $R_{S_{s, t}'}^\mathrm{mid} (B_{s, t}^*)$ can be obtained by a pricing derived from $S_{s, t}'$ by leaving the price of edges in the middle row unaltered, and setting the price of all other edges (except those priced $\infty$) to 0. Hence, it suffices to look for a pricing $S_{s, t}$ with the following shape: edges priced $\infty$ in $S_{s, t}'$ (black dotted) are priced $\infty$, all other edges (blue and red) are priced 0, except for one edge in the middle row in the $s$-$t$ range (green), whose price $p$ can vary (except in the case $s = t$, which is immediate). Optimizing over such pricings hence amounts to optimizing over $p$. This is straightforward to achieve in polynomial time by noting that only prices $p = b$ where $b$ is the budget of some driver $(u, v, b) \in B_H$ need to be considered, as otherwise, we could increase $p$ by $\epsilon > 0$ and not price anyone out of the market, increasing the revenue in the process. This achieves our goal for $\beta = 1$.
    \item For $i = \mathrm{low}$, note that a revenue of at least $R_{S_{s, t}'}^\mathrm{low} (B_{s, t}^*)$ can be obtained by a pricing derived from $S_{s, t}'$ by leaving the price of edges set to $\infty$ in $S_{s, t}'$ (black dotted) as $\infty$ and setting the price of all non-$\infty$ edges on the middle row and above to 0 (green and blue). Hence, it suffices to look for a pricing $S_{s, t}$ with the following shape: edges priced $\infty$ in $S_{s, t}'$ (black dotted) are priced $\infty$, other edges on or above the middle row are priced $0$ (green and blue), and the remaining edges (meaning those strictly below the middle row together with the edge going down from $t$ -- marked red) can be priced arbitrarily. Optimizing over such pricings amounts to solving a rooted instance consisting of the partial subgrid (red) of $H_\mathrm{ext}$ starting from the middle row of $H$ and going down, with the edges set to infinity removed. In this instance, the root is vertex $t$, and we replace for all drivers their pair $(u, v)$ with $(t, v)$. We know how to obtain a pricing that \Crounding-approximates the optimal revenue for rooted instances in polynomial time by \cref{sect:rooted}, meaning that we can achieve our goal for $\beta = \frac{1}{\Crounding}$.
\end{enumerate}

Combined, our considerations give the following result for a single block:

\begin{lemma}
For the instance defined by drivers $B_H$ and the extended block $H_{\mathrm{ext}}$, let $\mathrm{REV}_H$ be the revenue of the solution returned by the described algorithm and 
$\mathrm{OPT}_H$ be the revenue of an optimal solution. Then:
\[ \mathrm{REV}_H \geq \frac{1}{\Crounding \cdot 3 \cdot \omega^2} \mathrm{OPT}_H \]
\end{lemma}

Plugging 
back into our main algorithm, we gain another factor of $2$ from the odd-even splitting on each level, and a logarithmic factor from the decomposition, so our algorithm 
achieves an approximation factor of $24 \omega^2 \log m$ in poly-time for any fixed 
$\omega$.

\section{Conclusion and Future Work}
Having established our result for grids with bounded width, it would be interesting to see if our ideas extend to the case of bounded-pathwidth or bounded-treewidth graphs, or to grids with some edges removed. Moreover, solving grids without a bound on their width seems like a natural first step to understanding the complexity of our problem on real-world networks (e.g., Manhattan-like cities). Moreover, note how for $\omega = 1$ our problem admits a PTAS, but already for $\omega = 2$, all we could give was a logarithmic-factor approximation. It would be interesting to see if constant-factor approximations are possible, at least for small $\omega$.

More broadly, in the introduction, we noted a number of modeling challenges that need addressing before results in this line of work could become practically applicable. There is uncertainty in travelers' behavior, uncertainty in the budgets, and congestion that needs to be kept in check in a realistic situation. Furthermore, the strategic or otherwise non-rational behavior of agents should also be considered. Along the same lines, the assumption that drivers take the cheapest path without regard for actual traveling time is unrealistic and would need to be relaxed. It would be interesting to see which assumptions of our setting could be relaxed to bring practical applicability into this line of work.

{
\bibliographystyle{alphaurl}
\bibliography{bibliography}
}

\newpage
\appendix

\section{Single Price Algorithm}
\label{sec:single_price}

For this appendix, assume that $G = (V, E)$ is an arbitrary undirected graph. Consider the following simple mechanism for pricing the edges of $G$: choose a single price $p$ and price all edges as $p$. The following lemma shows that choosing $p$ appropriately leads to a logarithmic-factor approximation of the optimal revenue.

\begin{lemma}
There exists a single price $p$ such that if each edge
in the graph $G$ is priced at $p$, then the revenue is at least $\frac{1}{4 (\log |E| + \log |B| + 1)}$ of the optimal one.
\end{lemma}
\begin{proof}
    By $\mathrm{hop}(u, v)$ we denote the hop distance (the number of edges in the shortest path) between vertices $u$ and $v$.
    If all edges have the same cost, then $\frac{b}{\mathrm{hop}(u, v)}$ is the maximum price a driver $(u, v, b) \in B$ is willing to pay.
    It is because under a single price she will always choose a path consisting of $\mathrm{hop}(u, v)$ edges.
    Thus, we will group the drivers into $\log |E| + \log |B| + 1$ groups by similar values of $\frac{b}{\mathrm{hop}(u, v)}$.
    
    However, before doing so, we discard all drivers with $\frac{b}{\mathrm{hop}(u, v)} < \frac{b_{max}}{|E| |B|}$,
    where $b_{max}$ is the maximum budget of a driver.
    Observe that the sum of their budgets it at most $b_{max}$.
    Since drivers who have budget $b_{max}$ do not belong to this group (and there is at least one of them), the sum of budgets of all drivers
    with $\frac{b}{\mathrm{hop}(u, v)} < \frac{b_{max}}{|E| |B|}$ is at most half of the sum of budgets overall.

    Now, we partition the remaining drivers by the value of $\frac{b}{\mathrm{hop}(u, v)}$
    into $\log |B| + \log |E| + 1$ buckets of the form $\left( \frac{b_{max}}{2^{k+1}}, \frac{b_{max}}{2^k}\right]$ for
    $k \in \set{0, 1, \ldots, \floor{\log |E| + \log |B|}}$.
    Together, all those groups of drivers dispose of at least half of the total budget.
    By pigeonhole principle, the sum of the budgets of drivers in one of the buckets must be at least a
    $\frac{1}{2(\log |E| + \log |B| + 1)}$ fraction of the total budget overall.
    Let us fix such a bucket with $\frac{b}{\mathrm{hop}(u, v)} \in \left( \frac{b_{max}}{2^{k+1}}, \frac{b_{max}}{2^k}\right]$
    and denote respective drivers as $B_k$.
    Let us consider the revenue $\mathrm{REV}(p)$ for the single price $p = \frac{b_{max}}{2^{k+1}}$:
    \begin{gather*} \mathrm{REV}\bra{\frac{b_{max}}{2^{k+1}}}
    \geq \sum_{(u, v, b) \in B_k} \frac{b_{max}}{2^{k+1}} \cdot \mathrm{hop}(u, v) \geq
    \\ \frac{1}{2} \sum_{(u, v, b) \in B_k} \frac{b}{\mathrm{hop}(u, v)} \cdot \mathrm{hop}(u, v)
    = \frac{1}{2} \sum_{(u, v, b) \in B_k} b
    \end{gather*}

    The first inequality follows from the fact that all drivers in $B_k$ are able to afford their paths under $p = \frac{b_{max}}{2^{k+1}}$,
    the second follows from the definition of the driver partition.
    Since, as we have argued before, $\sum_{(u, v, b) \in B_k} b \geq \frac{1}{2(\log |E| + \log |B| + 1)} \sum_{(u, v, b) \in B} b$,
    this ends the proof, because the sum of budgets is a natural upper bound on the revenue.
\end{proof}

As we observed in the proof above, $\frac{b}{\mathrm{hop}(u, v)}$ is the maximal single price acceptable for a driver $(u, v, b) \in B$.
Thus, the optimal single price has to belong to $\longset{\frac{b}{\mathrm{hop}(u, v)}}{(u, v, b) \in B}$.
Otherwise, we could increase the price by a small $\epsilon > 0$ without losing any drivers, hence increasing the revenue in the process.
Consequently, the optimal single price can be found in polynomial time by checking all elements of the above set and
choosing the one that maximizes revenue.
Hence, we have proven:

\begin{restatable}{theorem}{singleprice}\label{thm:single_price}
There exists a polynomial-time approximation algorithm on general graphs achieving a revenue within a factor of $ \OO{\log |B| + \log |E|}$ of the optimal.
\end{restatable}

\section{Grid Graph Compression}
    \label{apx:compression_lemma}

    Here we provide proofs which were omitted from \cref{sec:decomp} due to space considerations.

    \compressionlemma*

    Before proceeding with the proof, we introduce the concept of crossing vertices and an auxiliary lemma.
    A \emph{crossing vertex} of two paths $P_1$ and $P_2$ is an endpoint of a maximal path shared by $P_1$ and $P_2$.
    Whenever we call a vertex crossing without referencing any specific pair of paths,
    we mean that it is a crossing vertex of some pair from a given set of paths.
    
    \crossingvtexlemma*
    \begin{proof}
        \newcommand{\cross}[2]{\mathrm{cross}\bra{#1, #2}}
        Let us consider any collection of shortest paths in $G$ between pairs of vertices in $S$, namely
        $ \longset{P_{a, b}}{a, b \in S, a \neq b} $,
        where $P_{a, b}$ is an arbitrarily chosen shortest $a$-$b$ path in $G$.
        Without loss of generality, we assume that these paths are simple (any path can be made simple by removing loops).
        We'll now present a procedure that limits the number of crossing vertices outside $S$ to
        ${\abs{S} \choose 2} \cdot \bra{{\abs{S} \choose 2} - 1}$
        without changing the distances between vertices in $S$.

        \begin{algorithmic}

    \State $\mathcal{R} \gets \emptyset$
    \For{$P \gets P_{a,b} \in \longset{P_{a,b}}{a, b \in S, a \neq b} $}
        \For{$R \in \mathcal{R}$}
            \If{$R$ and $P$ have more than two crossing vertices}
            \State $c_1, c_2 \dots c_q$ \Comment{Crossing vertices of $P$ and $R$}
            \State $P' \gets P[a \dots c_1] + R[c_1 \dots c_q] + P[c_q \dots b]$
            \State $P \gets P'$
            \EndIf
        \EndFor
        \State $\mathcal{R} \gets \mathcal{R} \cup \set{P}$
    \EndFor
\end{algorithmic}
    
        The above procedure gradually creates $\mathcal{R}$ -- a family of shortest paths as in
        $G$ with limited number of crossing vertices.
        Throughout the procedure we maintain that there are only at most
        $\abs{\mathcal{R}} \bra{ \abs{\mathcal{R}} - 1}$ crossing vertices that don't belong to $S$.
        With $\cross{P}{Q}$ we denote the set of crossing vertices of paths $P$ and $Q$.
        Then, we write our invariant as:
        \begin{equation}
            \label{crossing_invariant}
            \abs{ \bigcup_{R_1, R_2 \in \mathcal{R}} \cross{R_1}{R_2} \setminus S } \leq \abs{\mathcal{R}} \bra{ \abs{\mathcal{R}} - 1}
        \end{equation}

    Let us process the paths from $ \longset{P_{a, b}}{a, b \in S, a \neq b} $ in any fixed order.
    Each path $P$ will be iteratively rerouted with respect to each path already in $\mathcal{R}$
    to ensure that it is contributing at most two crossing vertices with each $R \in \mathcal{R}$.
    Let us consider a single such rerouting step for $P$, a shortest $a$-$b$ path, and $R$ -- a path that was already processed before.
    If $P$ and $R$ have at most two crossing vertices -- we leave $P$ as is and proceed with the next path $R' \in \mathcal{R}$.
    Otherwise, we create a new path $P'$ from $P$ by replacing the $c_1$-$c_q$ part with the $c_1$-$c_q$ part of $R$.
    Since $R$ also is a shortest (=cheapest) path, $R[c_1 \dots c_q]$ also is a cheapest $c_1$-$c_q$ path,
    making $P'$ a cheapest $a$-$b$ path.
    Note that only $c_1$ and $c_q$ are crossing vertices of $P'$ and $R$.

    Now we also need to show that no other new crossing vertices between $P'$ and the already processed paths from $\mathcal{R}$ were created.
    More formally, for each $R' \in \mathcal{R}$ that was already processsed with $P$:
    \begin{equation*}
        \cross{P'}{R'} \setminus \set{a,b} \quad \subseteq \quad \cross{P}{R'} \cup \cross{R}{R'} \cup \set{c_1, c_q}
    \end{equation*}
    Note that all crossing vertices on the right-hand side apart from $c_1$ and $c_q$ where already added to the set of crossing vertices
    before $P$ was rerouted with respect to $R'$.
    To prove this, we will consider all vertices lying strictly inside the three parts of $P'$:
    $P[a \dots c_1]$, $R[c_1 \dots c_q]$, and $P[c_q \dots b]$.    
    $P[a \dots c_1]$ and $P[c_q \dots b]$ did not change, so any crossing vertices have already been accounted for
    when $P$ was being rerouted with respect to other paths in $\mathcal{R}$:
    \[ \cross{P[a \dots c_1]}{R'} \setminus \set{a, c_1} \quad \subseteq \quad \cross{P}{R'} \]
    \[ \cross{P[c_q \dots b]}{R'} \setminus \set{b, c_q} \quad \subseteq \quad \cross{P}{R'} \]
    $R[c_1 \dots c_q]$ is already a subpath of $R$,
    so naturally all crossing vertices strictly inside it are also crossing vertices of $R$ and $R'$:
    \[ \cross{R[c_1 \dots c_q]}{R'} \setminus \set{c_1, c_q}
        \quad \subseteq \quad \cross{R}{R'} \]
    Thus, all crossing vertices of $P'$ and any $R'$ that was already processed (apart from $c_1$ and $c_q$)
    were already crossing vertices before.
    Hence, we have that $P$ contributes only $c_1$ and $c_q$ as new crossing vertices when rerouted with respect to $R$.

    In the end we replace $P$ with $P'$ and continue rerouting it with respect to the remaining paths in $\mathcal{R}$.
    In the end, when $P$ is added to $\mathcal{R}$, it produces at most two crossing vertices with each path already in $\mathcal{R}$.
    This maintains the invariant that there are at most $\abs{\mathcal{R}} \bra{ \abs{\mathcal{R}} - 1}$
    crossing vertices between the paths in $\mathcal{R}$ (excluding $S$, the path endpoints).
    Also $P$ remains a shortest $a$-$b$ path.
    Thus, because $\abs{\mathcal{R}} = {\abs{S} \choose 2}$, $\mathcal{R}$ is the desired collection of paths.
    \end{proof}

    \begin{proof}[Proof of \cref{lemma:compression}]
    With the help of the above lemma, we will now prove \cref{lemma:compression}.
    For ease of exposition in the proof we will deal with incomplete grids, i.e., grids where some edges are missing.
    We model those missing edges by setting their weights to $\infty$.
    This way we will reason about incomplete grids, but, in reality, the underlying grid will be complete.

    Starting with the original grid $G$, we prove the above lemma by creating consecutive graphs
    and showing that the distances between vertices in the first row (denoted $S$) are preserved in each step.

    \paragraph{\textbf{$G'$ -- shortest path graph with few crossing vertices}}
    By \cref{lemma:crossing_vertices} there exists a collection $\mathcal{R}$ of shortest paths between vertices in $S$
    that results in at most ${\omega \choose 2} \cdot \bra{{\omega \choose 2} - 1}$ crossing vertices outside of the first row.
    We define $G'$ as the union of all paths in $\mathcal{R}$.
    Note that, by construction, pairwise distances between vertices in $S$ are the same in $G'$ as in $G$.

    \paragraph{\textbf{$G''$ -- a grid with small depth}}
    Now, let us compress $G'$ to a grid of depth at most $\MaxGadgetDepth$.
    Let us consider maximal ranges of consecutive rows in $G' \setminus S$ that only contain vertices of degree $2$.
    For brevity we will call them $2$-layers.
    
    Let us consider a single $2$-layer $L$.
    It must be a $l \times \omega$ (partial) grid for some $l$.
    If $l \leq \omega+1$, we leave $L$ as is.
    Otherwise, we will compress it to a $\omega \times (\omega+1)$ (partial) grid.
    Let $U$ and $D$ be  the sets of vertices in the top and bottom row of $L$
    that have edges outgoing from $L$.
    Note that $U \cup D$ are exactly the vertices having edges outgoing from $L$.
    Thus, to preserve distances between vertices in $S$ globally, it is enough to maintain the distances from $L$ between vertices in $U \cup D$.
    Now let us look at $L$ in isolation from the rest of the graph.
    
    \begin{figure}
        \centering
        \vspace{-0.5cm}
        \includegraphics[width=0.5\textwidth]{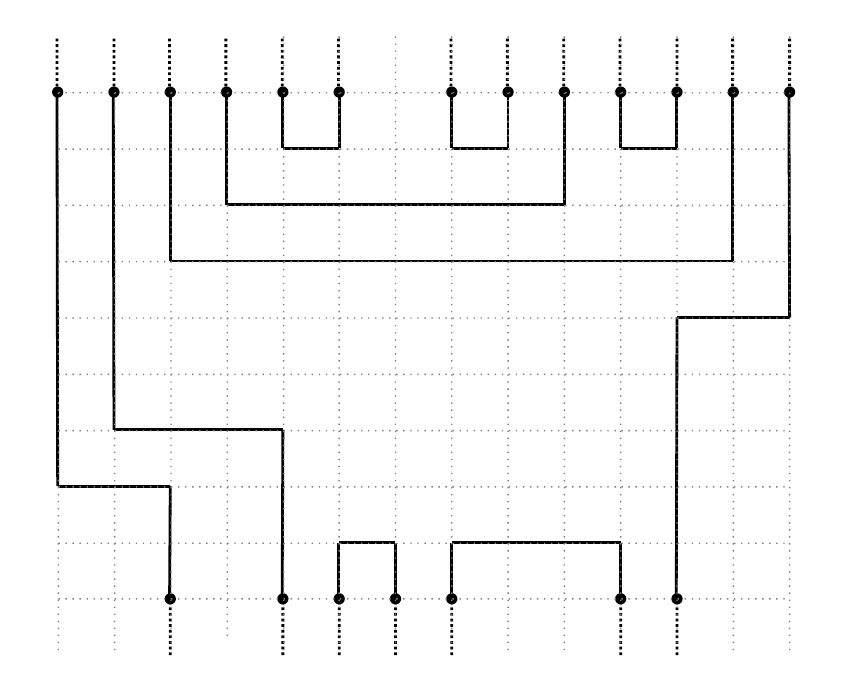}
        \vspace{-0.5cm}
        \caption{
            An example of a compressed $2$-layer, where vertices in $U \cup D$ are connected using the newly created paths.
            Note that not all vertices in the border rows belong to $U$ or $D$,
            but only those that have edges going outside of the $2$-layer.
        }
        \label{fig:nestedgrid}
    \end{figure}

    Note that inside $L$ (in isolation from the rest of the graph), all vertices have degree two except for those in $U \cup D$.
    This is because all $L$'s vertices had degree two in $G''$ and of those in $U \cup D$ each had one edge outgoing from $L$.
    Consequently, $L$ can be partitioned into paths that connect vertices in $U \cup D$.
    There will be no cycles in this decomposition, because $G''$ is a union of paths starting and ending outside of $L$.
    Again, because the maximal degree is two, those path do not cross (are vertex disjoint).
    Thus, each vertex $a \in U \cup D$ is connected (via such a path) to exactly one vertex $b$ in $U \cup D$
    (it is a perfect matching on $U \cup D$).
    In our compressed grid $L'$ we will connect each such $a$ with the matching $b$ using vertex disjoint paths.
    Then, for each $a, b \in U \cup D$ we can make the total costs along such an $a$-$b$ path equal to
    the corresponding distance from $a$ to $b$ in $L$
    (for example by setting one edge's weight to it and leaving the other weights at $0$).
    Now, we will show that those disjoint paths can be created in a grid of depth $\omega + 1$.
    
    With $U_{in}$ we term the set of vertices in $U$ that are matched with other vertices in $U$.
    Let us consider all vertices in $U$ in the order from left to right as they appear in $L$.
    Let $a, b \in U_{in}$ be connected by a path in $L$.
    We will show by contradiction that any vertex $c \in U$ which is between $a$ and $b$
    must be matched with a vertex $d \in U$ that also lies between $a$ and $b$.
    If $d \notin U$ ($d \in D$), any $a$-$b$ path crosses with any $c$-$d$ path,
    because the latter would have to reach the bottom row from the topmost one of $L$ and the former starts and ends in the topmost row.
    Such a crossing would necessarily result in a vertex of degree at least $3$, which, by the definition of a $2$-layer, is impossible.
    If $d \in U$, but was not between $a$ and $b$, each possible pair of $a$-$b$ and $c$-$d$ paths would also have to cross at least once,
    which is impossible for the same reason.
    Thus, vertices in $U_{in}$ can be connected as shown in \cref{fig:nestedgrid}.

    Since the maximal nesting level of the connected pairs is $\frac{\abs{U_{in}}}{2}$ and one extra row of vertices is needed per level of nesting,
    this gives us a depth of $\frac{\abs{U_{in}}}{2} + \frac{\abs{D_{in}}}{2}$.
    Of course, vertices in the analogous set $D_{in}$ are connected in the same way.
    
    Now, each vertex in $U_{out} := U \setminus U_{in}$ is matched with a vertex from $D_{out} := D \setminus D_{in}$.
    Note that this perfect matching between $U_{out}$ and $D_{out}$ preserves the left-to-right ordering of the vertices
    (that is, the leftmost vertex in $U_{out}$ is matched with the leftmost one in $D_{out}$ and so on).
    It is because if we had two pairs of matched vertices violating this order, the paths connecting them would cross.
    In the previous stage, no edges have been added in the columns of vertices from $U_{out}$ and $D_{out}$, so we add respectively $\frac{\abs{U_{in}}}{2} + 1$ and $\frac{\abs{D_{in}}}{2} + 1$ vertical edges outgoing from them.
    Consequently, the algorithm keeps prolonging the resulting paths starting in $U_{out}$ and $D_{out}$ by appending edges to the other end until all matched pairs of vertices are connected. This happens from left to right. At each step the algorithm selects vertices $u$ and $d$, the ends of paths starting in leftmost of the yet unprocessed vertices from respectively $U_{out}$ and $D_{out}$, and connects them according to the following procedure:
    \begin{enumerate}[nosep]
    \item From whichever of $u$ and $d$ is further right, lead a horizontal path to the left until the column of the other one.
    \item Connect the paths ending in $u$ and $d$ using a vertical path of length $\abs{U_{out}}$ (twice the number of pairs left to connect including the current one).
    \item Remove vertices corresponding to $u$ and $d$ from $U_{out}$ and $D_{out}$.
    Depending on whether $u$ or $d$ was further right, prolong the vertical paths outgoing from vertices in $U_{out}$ ($u$) or $D_{out}$ ($d$) by one edge.
    \end{enumerate}
    
    By construction we are guaranteed that in each iteration the first edge to the left of $u$ and $d$ is
    the vertical path added in the step two of the previous iteration,
    which must be to the left of both $u$ and $d$.
    Hence, the horizontal path from the first step will not cross with any path.
    Since $\abs{U_{out}} = \abs{D_{out}}$, the above procedure connects exactly the matched pairs from $U_{out} \times D_{out}$.
    It needs extra depth of $\abs{U_{out}} + 1$.
    
    The above procedure creates a grid where exactly the matched pairs of vertices are connected and the matching is the same as in $L$.
    The resulting depth of $L'$ is equal to
    $\frac{\abs{U_{in}}}{2} + \frac{\abs{D_{in}}}{2} + 1 + \abs{U_{out}}
    \leq \max \bra{ \abs{U}, \abs{D}} + 1$.
    Thus, $L$ can be compressed to a $(\omega+1) \times \omega$ grid.
        
    Since there are at most ${\omega \choose 2} \cdot \bra{{\omega \choose 2} - 1}$ crossing vertices outside the first row,
    there can be at most ${\omega \choose 2} \cdot \bra{{\omega \choose 2} - 1}$ rows in the grid $G'$ that contain a vertex
    of degree $3$ or more (let us color those rows black).
    Now $G'$ consists of the row formed by vertices in $S$, the black rows, and $2$-layers.
    Let us create $G''$ from $G'$ by compressing the $2$-layers of more than $\omega + 1$ rows as described earlier.
    Now, in $G''$ we have at most ${\omega \choose 2} \cdot \bra{{\omega \choose 2} - 1}$ black rows,
    with each pair of consecutive ones being separated by at most $\omega+1$ rows.
    Adding the one row for $S$, this gives us
    ${\omega \choose 2} \cdot \bra{{\omega \choose 2} - 1} \cdot (\omega+2) + (\omega + 1) + 1 \leq \MaxGadgetDepth$
    rows in the grid $G''$,
    which has the same distances between vertices in $S$ as $G$.
    \end{proof}

%%%%%%%%%%%%%%%%%%%%%%%%%%%%%%%%%%%%%%%%%%%%%%%%%%%%%%%%%%%%%%%%%%%%%%%%

\end{document}